\documentclass[12pt]{amsart}
\usepackage{amssymb}
\usepackage{color}
\usepackage[normalem]{ulem}
\usepackage{epic,curves}
\usepackage[english]{babel}
\usepackage{graphicx}
\usepackage[all]{xy}
\usepackage{multirow}
\usepackage{enumerate}
\usepackage{comment}
\newtheorem{claim}{}[section]
\newtheorem{theorem}[claim]{Theorem}

\newtheorem{proposition}[claim]{Proposition}
\newtheorem{corollary}[claim]{Corollary}
\newtheorem{example}[claim]{Example}

\renewenvironment{proof}{\noindent{\it Proof. \hskip0pt}}
                      {$\square$\par\medskip}
\newcounter{question}
\newenvironment{question}{\par\medskip\par\parindent=-3mm\
    \refstepcounter{question}
    \bf\sf{ Question \thequestion}  \parindent=3mm\rm \
                }{\par\medskip}

\begin{document}
\baselineskip 6.0 truemm
\parindent 1.5 true pc

\newcommand\lan{\langle}
\newcommand\ran{\rangle}
\newcommand\tr{{\text{\rm Tr}}\,}
\newcommand\ot{\otimes}
\newcommand\ol{\overline}
\newcommand\join{\vee}
\newcommand\meet{\wedge}
\renewcommand\ker{{\text{\rm Ker}}\,}
\newcommand\image{{\text{\rm Im}}\,}
\newcommand\id{{\text{\rm id}}}
\newcommand\Sym{{\text{\rm Sym}}}
\newcommand\tp{{\text{\rm tp}}}
\newcommand\pr{\prime}
\newcommand\e{\epsilon}
\newcommand\la{\lambda}
\newcommand\inte{{\text{\rm int}}\,}
\newcommand\ttt{{\text{\rm t}}}
\newcommand\spa{{\text{\rm span}}\,}
\newcommand\conv{{\text{\rm conv}}\,}
\newcommand\rank{\ {\text{\rm rank of}}\ }
\newcommand\re{{\text{\rm Re}}\,}
\newcommand\ppt{\mathbb T}
\newcommand\rk{{\text{\rm rank}}\,}
\newcommand\bcolor{\color{blue}}
\newcommand\ecolor{\color{black}}
\newcommand\sss{\omega}
\newcommand\sgn{{\text{\rm sgn}}\,}

\newcommand{\CC}{\mathbb{C}}
\def\Cdn{\bigotimes_{j=1}^n \CC^{d_j}}
\def\Otn{\bigotimes_{i=1}^n}

\newcommand{\bra}[1]{\langle{#1}|}
\newcommand{\ket}[1]{|{#1}\rangle}
\newcommand{\ZZ}{\mathbb{Z}}
\newcommand{\seq}[2]{ {#1}_1 ,{#1}_2 , \cdots, {#1}_{#2} }
\newcommand{\PP}{\mathbb{P}}
\def\cU{{\mathcal U}}
\newcommand{\RR}{\mathbb{R}}
\newcommand\Gr{{\text{\rm Gr}}\,}
\newcommand{\bk}{\mathbf{k}}
\newcommand{\bm}{\mathbf{m}}
\newcommand{\bv}{\mathbf{v}}
\newcommand{\bw}{\mathbf{w}}
\newcommand\cT{{\mathcal T}}
\newcommand{\be}{\mathbf{e}}
\def\pri{^{\prime}}
\newcommand\per{\text{\rm per}}
\newcommand{\braket}[2]{\langle{#1}|{#2}\rangle}
\def\orth{^{\perp}}
\def\ts{\otimes}
\newcommand{\floor}[1]{\left\lfloor{#1}\right\rfloor}

\newcommand{\del}[1]{\textcolor[rgb]{0.98,0.5,0}{\sout{#1}}}
\newcommand{\add}[1]{\textcolor[rgb]{1,0,0}{#1}}

\title{Product vectors in the ranges of multi-partite states with positive partial transposes and permanents of matrices}

\author{Young-Hoon Kiem}
\address{Department of Mathematics and Research Institute of Mathematics\\ Seoul National University\\ Seoul 151-747, Korea}
\email{kiem@snu.ac.kr}
\author{Seung-Hyeok Kye}
\address{Department of Mathematics and Research Institute of Mathematics\\ Seoul National University\\ Seoul 151-747, Korea}
\email{kye@snu.ac.kr}
\author{Joohan Na}
\address{School of Computational Sciences\\ Korea Institute for Advanced Study\\ Seoul 130-722, Korea}
\email{jhna@kias.re.kr}

\date{May 10, 2015}

\thanks{YHK and JN were partially supported by NRF grant 2011-0027969; SHK was partially supported by NRFK grant 2013-004942}



\begin{abstract}
In this paper, we consider a system of homogeneous algebraic equations in complex variables and their conjugates which arise naturally from the range criterion for separability of PPT states. We examine systematically these equations to get sufficient conditions for the existence of nontrivial solutions. This gives us possible upper bounds of ranks of PPT entangled edge states and their partial transposes. We will focus on the multi-partite cases which are much more delicate than the bi-partite cases. We use the notion of permanents of matrices as well as techniques from algebraic geometry through the discussion.
\end{abstract}

\maketitle

\section{Introduction}\label{sec:intro}

Quantum entanglement is now considered as the main resource for quantum information and quantum computation, and distinguishing entanglement from separability is one of the most important problems in the theory. The most convenient and powerful criterion is the so called PPT criterion by Choi \cite{choi-ppt} and Peres \cite{peres}, which tells us that the partial transposes of a separable state are positive, that is, positive semi-definite. But, it is very difficult to determine if a given PPT state is separable or not, and it is actually known to be an $NP$-hard problem \cite{StrongNP, NPHard03}. In order to determine if a given PPT state is separable, it is natural to look at the ranges of the state and its partial transposes, as it had been suggested by the range criterion  \cite{p-horo}. Apart from the separability criterion, the notion of PPT is interesting in itself as well. For example, it is closely related with the question of distillability, which is one of the main open problems in quantum information theory. See \cite{horo-distill}.

Special kinds of PPT entangled states, PPT edge states, play an important role to understand the whole structures of PPT states, because every PPT entangled state is the sum of a separable state and a PPT edge state. See \cite{lew00}. Because PPT edge states have typically low ranks, it is very important to get upper bounds for possible ranks of PPT entangled edge states and their partial transposes. To do this, we need to consider a system of algebraic equations arising from the range criterion, as it was initiated in \cite{2xn} for the bi-partite $2\otimes n$ cases. The main purpose of this note is to provide a rigorous mathematical background to deal with these equations in general multi-partite cases,
using techniques from algebraic geometry.

A (mixed) state on the Hilbert space ${\mathcal H}=\Cdn$ is a positive semi-definite Hermitian operator of trace one. Throughout this note, we assume that a state always means a mixed state. A state on the Hilbert space ${\mathcal H}=\Cdn$ is said to be \emph{separable} if it is a convex combination of pure product states of the form
$$
\ket{\psi_1}\bra{\psi_1}\otimes \ket{\psi_2}\bra{\psi_2}\otimes\cdots \otimes \ket{\psi_n}\bra{\psi_n} \in M_{d_1}\ot M_{d_2}\ot \cdots \ot M_{d_n},
$$
where $M_d$ denotes the algebra of all $d\times d$ matrices over the field of complex numbers. A state is called \emph{entangled} if it is not separable. For a given subset $S$ of $[n]:=\{1,2,\cdots,n\}$, we define the \emph{partial transpose} $\left( \bigotimes_{j=1}^n A_j \right)^{T(S)}$ of $\bigotimes_{j=1}^n A_j$ by
$$
(A_1\ot A_2\ot\cdots\ot A_n)^{T(S)}:=B_1\ot B_2\ot\cdots\ot B_n, \quad \text{\rm with}\ B_j=\begin{cases} A_j^\ttt, &j\in S,\\ A_j, &j\notin S,\end{cases}
$$
and extend the map to the whole $\bigotimes_{j=1}^n M_{d_j}$ by linearity, where $A^\ttt$ denotes the transpose of the matrix $A$. We say that a state $\varrho$ is of PPT if its partial transpose $\varrho^{T(S)}$ is positive for every subset $S$ of $[n]$. It is easily checked that every separable state is of PPT, as it was observed by Choi \cite{choi-ppt} and Peres \cite{peres} for the bi-partite case $n=2$. We note that $\varrho^{T(S)}$ is positive if and only if $\varrho^{T(S^c)}$ is positive, where $S^c$ is the complement of $S$ in $[n]$. Therefore, it is enough to check the positivity of $2^{n-1}$ matrices among $2^n$ matrices, to confirm the PPT property of a given $n$-partite state.

For a subset $S$ of $[n]$ and a product vector $\ket\psi=\ket{\psi_1}\ot\cdots\ot\ket{\psi_n}$, we define the product vector $\ket\psi^{\Gamma(S)}$ up to constant by
\begin{equation}\label{par-conj}
(\ket{\psi_1}\ot\cdots\ot\ket{\psi_n})^{\Gamma(S)} :=\ket{\phi_1}\ot\cdots\ot\ket{\phi_n}, \quad \text{\rm with}\ \ket{\phi_j}=\begin{cases} \ket{\bar{\psi_j}}, &j\in S,\\ \ket{\psi_j}, &j\notin S.\end{cases}
\end{equation}
Note that the range of a density matrix $\sum_i|z_i\rangle\langle z_i|$ is given by the span of $\{|z_i\rangle\}$. Therefore, if a given PPT state $\varrho$ is separable then there exists a collection $\Psi$ of product vectors with the property \cite{p-horo} : The range of $\varrho^{T(S)}$ is the span of the product vectors $\{\ket\psi^{\Gamma(S)}:\ket\psi\in\Psi\}$ for each subset $S$ of $[n]$. Therefore, the first step to confirm separability of a PPT state $\varrho$ is to check the existence of a nonzero product vector $\ket\psi$ such that $\ket{\psi}^{\Gamma(S)}$ belongs to the range of $\varrho^{T(S)}$ for each subset $S$ of $[n]$. A PPT state $\varrho$ is said to be an edge state if there exist no such nonzero product vectors.

Suppose that we are given finite sequences $\{ S_1,S_2,\cdots S_r\}$ of subsets of $[n]$ and $\{ D_1,D_2\cdots, D_r\}$ of subspaces of $\mathcal H$. The purpose of this note is to investigate the system of equations
\begin{equation}\label{eq}
\ket\psi^{\Gamma(S_i)}\in D_i,\qquad i=1,2,\cdots, r,
\end{equation}
with unknowns $\ket\psi=\ket{\psi_1}\ot\cdots\ot\ket{\psi_n}$ in the product $\mathbb C\mathbb P^{d_1-1}\times \cdots\times \mathbb C\mathbb P^{d_n-1}$ of complex projective spaces, and find conditions for which the system of equations (\ref{eq}) has a nonzero solution, i.e. we are looking for nonzero solutions up to nonzero scalar multiplication. These equations have been considered in earlier papers \cite{hlvc}, \cite{PPT3qubit} and \cite{2xn}, for examples. See also \cite{kye-prod-vec} for the bi-partite case $n=2$.

To do this, it is convenient to define the $r\times n$ matrix $\Sigma=[\sigma_{ij}]$ with entries
$$
\sigma_{ij}=\begin{cases} -1, &j\in S_i,\\ +1, &j\notin
S_i,\end{cases}
$$
which will be called the {\sl associated matrix} of the sequence $\{S_1,S_2,\cdots S_r\}$. We note that the number $N_U$ of unknowns of the system of equations (\ref{eq}) is given by
$$
N_U=\sum_{j=1}^n(d_j-1),
$$
which is the complex dimension of the manifold $\mathbb C\mathbb P^{d_1-1}\times \cdots\times \mathbb C\mathbb P^{d_n-1}$. On the other hand, the number $N_E$ of algebraic equations in (\ref{eq}) is just $N_E = \sum_{i=1}^r\dim D_i^\perp$. Now, we are ready to state the main result of this note:

\begin{theorem}\label{main}
Let $\{S_1,\cdots,S_r\}$ be a sequence of subsets of $[n]$ with the associated matrix $\Sigma=[\sigma_{ij}]$. Then we have the following:
\begin{enumerate}
\item[{\rm (i)}]
If $N_E = \sum_{i=1}^r k_i > N_U$, then the system of equations {\rm (\ref{eq})} has no nonzero solution for generic subspaces $D_i$ of $\mathcal H$ with $k_i=\dim D_i^{\perp}$ for $i=1,2,\cdots,r$.
\item[{\rm (ii)}]
Suppose that $N_E=N_U$, and the coefficient of $\prod_j\alpha_j^{d_j-1}$ is nonzero when we expand the polynomial
\begin{equation}\label{criti}
{\prod_{i=1}^r(\sigma_{i,1} \alpha_1 + \sigma_{i,2} \alpha_2 + \cdots + \sigma_{i,n} \alpha_n)^{\dim D_i^\perp}}_.
\end{equation}
Then the system of equations {\rm (\ref{eq})} has a nonzero solution.
\item[{\rm (iii)}]
If $N_E < N_U$ and the associate matrix $\Sigma$ has rank $r$, then the system of equations {\rm (\ref{eq})} has infinitely many solutions.
\end{enumerate}
\end{theorem}

The case (ii) considers the critical case, in which the numbers of unknowns and equations coincide. We have exhibited in \cite{kye-prod-vec} examples in the two qubit and the two qutrit cases for which the system of equations (\ref{eq}) has no nonzero solution even though it has the same numbers of unknowns and equations. We have the same kind of an example for the four qubit case. See Example \ref{4_qu}. In the $n$ qubit cases, the condition in (ii) can be checked by computing the permanent of the associated matrix $\Sigma=[\sigma_{i,j}]$, which may be assumed to be a square matrix. It was shown in \cite{per05} that the permanent of an $n\times n$ square matrix whose entries are $\pm 1$ is never zero  if and only if $n=2^k-1$ with $k=2,3,\cdots$. From this, we may conclude that the coefficient condition in (ii) is redundant for the $n$ qubit cases, whenever $n=2^k-1$ with $k=2,3,\cdots$.  This is especially true for the three qubit case. The notion of permanent is also useful in the theory of entanglement in other contexts. See \cite{wei}.

The cases (i) and (iii) deal with the over-determined and under-determined cases, respectively. We need an unexpected rank condition to get the existence of infinitely many solutions for the under-determined case. We do not know if one may remove this condition, even though we provided an example which strongly suggests the role of rank condition. See Example \ref{ex_va}. We will see that the rank condition is redundant for the bi-partite case, the three and four qubit cases.

In the three qubit case, it is enough to consider four subsets $S_1=\{1\}, S_2=\{2\}, S_3=\{3\}, S_4=\emptyset$ of $[3]=\{1,2,3\}$, to check the PPT condition. The above discussions with the statements (ii) and (iii) prove the inequality
\begin{equation}\label{3qu}
\sum_{i=1}^4\rk \varrho^{T(S_i)}< 4\times 2^3-3=29
\end{equation}
for any three qubit PPT entangled edge states, as it was stated in \cite{abls}. It should be noted \cite{kye_osaka} that the corresponding inequality is false for the $3\otimes 3$ case. It is an open question if the corresponding inequality holds for the $2\otimes 4$ case. It is also unknown for the $n$ qubit cases when $n\ge 4$.

Theorem \ref{main} will be proved in the next section. We consider the multi-qubit cases in Section 3 to relate the notion of permanents and the existence of a nonzero solution of the system of equations (\ref{eq}). We also classify $4\times 4$ $(+1,-1)$-matrices with vanishing permanents, up to an equivalence relation. In the final section, we consider the problem of classifying PPT entangled edge states by their ranks of partial transposes, and discuss related questions.

%%%%%%%%%%%%%%%%%%%%%%%%%%%%%%%%%%%%%%%%%%%%%%%%%%%%%%%%%%%%%%%%%%%%
\section{Proof of Main Theorem}
%%%%%%%%%%%%%%%%%%%%%%%%%%%%%%%%%%%%%%%%%%%%%%%%%%%%%%%%
\subsection{Over-determined case}

We let $\Gr(d,k)$ denote the set of all subspaces $D$ of $\CC^d$ with $\dim D^{\perp}=k$, which is a manifold of complex dimension $k(d-k)$. See \cite[Chapter 1, Section 5]{GH_principles}. We say that a property holds for \emph{generic} subspaces $D$ in $\CC^d$ with $\dim D^{\perp}=k$ if there is a subset $\mathcal{W}$ of $\Gr(d,k)$ whose complement is of measure zero such that the property holds for all $D \in \mathcal{W}$.

Theorem \ref{main} (i) is a consequence of dimension estimates and the Morse-Sard theorem. We will write $\PP^{d-1}$ for the complex projective space $\CC\PP^{d-1}$ to simply the notations. Consider the following diagram:

\begin{equation}\label{diaphidef}
\xymatrix@M=12pt{
\prod_{j=1}^n \PP^{d_j-1} \ar[d]_{\txt<2pc>{$\mathit{\phi}^{\Gamma(S)}$}} \quad && \\
\prod_{j=1}^n \PP^{d_j-1}  \quad
\ar@{^{(}->}[rr]^{\txt<8pc>{$\mathit{\iota}$}} && \quad \PP^{d -1}}
\end{equation}
where ${d}=\prod_j d_j$. For $S\subset [n]$, the map $\phi^{\Gamma(S)}$ is the diffeomorphism which sends $([\ket{\psi_j}])$ to $([\ket{\phi_j}])$, where $\ket{\phi_j}$ is given by (\ref{par-conj}), and $[\ket{\psi}] \in \PP^{d-1}$ denotes the line spanned by $\ket{\psi}$. The injective map $\iota$ is the Segre embedding which sends $([\ket{\psi_j}])$ to $[\otimes_{j=1}^n \ket{\psi_j}]$. We want to show that the set
\begin{equation}\label{intexm}
\bigcap_{i=1}^r \left( \phi^{\Gamma(S_i)} \right)^{-1} \left( \iota^{-1} ( \PP D_i ) \right) = \left\{ ([\ket{\psi_j}]) \in
\prod_{j=1}^n \PP^{d_j -1}\, \mid \, \left( \otimes_{j=1}^n \ket{\psi_j} \right)^{\Gamma(S_i)} \in D_i \right\}
\end{equation}
is empty for generic choices of $D_i$.

By Bertini's theorem \cite[Chapter II, Theorem 8.18]{hart} in algebraic geometry, we may choose a generic $D_1$ such that $\iota^{-1} ( \PP D_1 )$ is a smooth manifold of real dimension $2(N_U-k_1)$. Let $E_1 := \left( \phi^{\Gamma(S_1)} \right)^{-1} \left( \iota^{-1} ( \PP D_1 ) \right)$. To choose $D_2$, let us consider the universal bundle $\cU_2$ over $\Gr({d},k_2)$ so that we have a diagram:

\[\xymatrix{ \PP\cU_2\ar[d]\ar@{^(->}[r] & \PP^{{d}-1}\times \Gr({d},k_2)\\ \Gr({d},k_2) }\]

Each fiber of the vertical arrow over a point $\xi \in \Gr({d},k_2)$ gives the linear subspace $\PP D_\xi\subset \PP^{{d}-1}$ represented by $\xi$. Via the Segre embedding, we can regard $\phi^{\Gamma(S_2)}(E_1)\times \Gr({d},k_2)$ as a subset of $\PP^{{d}-1}\times \Gr({d},k_2)$. Take the intersection $(\phi^{\Gamma(S_2)}(E_1)\times \Gr({d},k_2))\cap \PP \cU_2$. There are obvious projections 
\[\xymatrix{ & (\phi^{\Gamma(S_2)}(E_1)\times \Gr({d},k_2))\cap \PP \cU_2 \ar[dl]_p\ar[dr]^q\\ \phi^{\Gamma(S_2)}(E_1)&& \Gr({d},k_2) }\] 

Let us estimate the dimension of this intersection. For each point $\eta$ in $\phi^{\Gamma(S_2)}(E_1)$, $p^{-1}(\eta)$ is
$$\{D_2\in \Gr({d},k_2)\,|\, \eta\in \PP D_2\}\cong \Gr({d}-1,k_2)$$ 
since a subspace of $\CC^{{d}}$ of codimension $k_2$ containing a line $l_\eta$ represented by $\eta$ is uniquely determined by a subspace of $\CC^{{d}}/l_\eta=\CC^{{d}-1}$ of codimension $k_2$. Therefore, the intersection $(\phi^{\Gamma(S_2)}(E_1) \times \Gr({d},k_2))\cap \PP \cU_2$ is a smooth real manifold of real dimension
$$
2(N_U-k_1)+\dim_\RR \Gr({d}-1,k_2)=2(N_U-k_1)+2k_2({d}-1-k_2).
$$

If $q$ is not surjective, the fiber $q^{-1}(D_2) = \phi^{\Gamma(S_2)}(E_1) \cap \iota^{-1} (\PP D_2)$ is empty for a generic choice of $D_2 \in \Gr({d},k_2)$. Let $E_2 := \left( \phi^{\Gamma(S_2)} \right)^{-1} \left( \iota^{-1} ( \PP D_2 ) \right)$. Then
\begin{equation}\label{phiint}
\begin{split}
E_1 \cap E_2 = & E_1 \cap \left( \phi^{\Gamma(S_2)} \right)^{-1} \left( \iota^{-1} ( \PP D_2 ) \right) \\ = & \left( \phi^{\Gamma(S_2)} \right)^{-1} \left( \phi^{\Gamma(S_2)} (E_1) \cap  \iota^{-1} ( \PP D_2 ) \right) \\ = & \left( \phi^{\Gamma(S_2)} \right)^{-1} \left( q^{-1}(D_2) \right) =  \emptyset
\end{split}
\end{equation}
for such a generic $D_2$, so we have the statement (i). Thus we may assume that $q$ is surjective.

Applying the Morse-Sard theorem \cite[Chapter 3, Theorem 1.3]{hirsch} to the smooth map $q:(\phi^{\Gamma(S_2)}(E_1)\times
\Gr({d},k_2))\cap \PP \cU_2 \to \Gr({d},k_2)$, we find that over a generic choice of $D_2\in \Gr({d},k_2)$, the fiber $q^{-1}(D_2)$ of $q$ is a smooth manifold of real dimension
$$
2(N_U-k_1)+2k_2({d}-1-k_2) - 2k_2({d}-k_2)=2(N_U-k_1-k_2).
$$
For such a generic $D_2$, if we let $E_2 := \left( \phi^{\Gamma(S_2)} \right)^{-1} \left( \iota^{-1} ( \PP D_2 ) \right)$, then by \eqref{phiint}, $E_1\cap E_2 = \left( \phi^{\Gamma(S_2)} \right)^{-1} \left( q^{-1}(D_2) \right)$ is a smooth manifold of expected real dimension $2(N_U-k_1-k_2)$.

Now it is obvious how to proceed. We consider the universal bundle $\cU_3$ over $\Gr({d},k_3)$, the intersection
$$(\phi^{\Gamma(S_3)}(E_1\cap E_2)\times \Gr({d},k_3))\cap \PP \cU_3$$
and the projections to $\phi^{\Gamma(S_3)}(E_1\cap E_2)$ and $\Gr({d},k_3)$. If the projection to $\Gr({d},k_3)$ is not surjective, then $\phi^{\Gamma(S_3)}(E_1\cap E_2)\cap \iota^{-1} (\PP D_3 )$ is empty for a generic $D_3\in \Gr({d},k_3)$. For such a generic $D_3$, if we let $E_3 := \left( \phi^{\Gamma(S_3)} \right)^{-1} \left( \iota^{-1} ( \PP D_3 ) \right)$, 
$$E_1\cap E_2\cap E_3 = \left( \phi^{\Gamma(S_3)} \right)^{-1} \left( \phi^{\Gamma(S_3)} (E_1 \cap E_2) \cap  \iota^{-1} ( \PP D_3 ) \right)$$ 
is also empty and we have the theorem.

If the projection to $\Gr({d},k_3)$ is surjective, by the Morse-Sard theorem, we find that for a generic $D_3\in \Gr({d},k_3)$, $\phi^{\Gamma(S_3)}(E_1\cap E_2)\cap \PP D_3$ is a smooth manifold of real dimension $2(N_U-k_1-k_2-k_3)$. Then letting $E_3 := \left( \phi^{\Gamma(S_3)} \right)^{-1} \left( \iota^{-1} ( \PP D_3 ) \right)$, $E_1\cap E_2\cap E_3$ is also a smooth manifold of real dimension $2(N_U-k_1-k_2-k_3)$ for such a generic $D_3$. Continuing this way, the intersection 
$$\bigcap_{i=1}^r E_i = \bigcap_{i=1}^r \left( \phi^{\Gamma(S_i)} \right)^{-1} \left( \iota^{-1} ( \PP D_i ) \right)$$ 
eventually becomes empty for generic choices of $D_i$ since $N_U<\sum_i k_i=N_E$. This proves (i) of Theorem \ref{main}.

\subsection{Critical case}
For the statements (ii) and (iii), we need the following theorem which gives us an algebraic sufficient condition for the existence of nonzero solutions of the system of equations (\ref{eq}).

\begin{theorem}\label{algcond}
Let  $\{S_1,\cdots,S_r\}$ be sequences of subsets of $[n]$ and $\{D_1,\cdots,D_r\}$ subspaces of $\mathcal H=\bigotimes_{j=1}^n\mathbb C^{d_j}$ with $k_i=\dim\ D_i^\perp$. If
\begin{equation}\label{eq:1}
\prod_{i=1}^r (\sigma_{i,1} \alpha_1 + \sigma_{i,2} \alpha_2 + \cdots + \sigma_{i,n} \alpha_n)^{k_i} \neq 0
\end{equation}
in the ring $\ZZ[\seq{\alpha}{n}]/(\alpha_1^{d_1}, \alpha_2^{d_2}, \cdots, \alpha_n^{d_n}),$ then the system of equations {\rm (\ref{eq})} has a nonzero solution.
\end{theorem}

Here, $\alpha_1, \cdots,\alpha_n$ are indeterminates, and $\ZZ[\seq{\alpha}{n}]/(\alpha_1^{d_1}, \alpha_2^{d_2}, \cdots, \alpha_n^{d_n})$ denotes the quotient ring of the polynomial ring $\ZZ[\seq{\alpha}{n}]$ by the ideal generated by $\alpha_1^{d_1}, \alpha_2^{d_2}, \cdots, \alpha_n^{d_n}$.

\medskip

\begin{proof}
Consider the diagram \eqref{diaphidef}. We have to measure the size of the set
$$
\bigcap_{i=1}^r \left( \phi^{\Gamma(S_i)} \right)^{-1} \left(
\iota^{-1} ( \PP D_i ) \right) = \left\{ ([\ket{\psi_j}]) \in
\prod_{j=1}^n \PP^{d_j -1}\, |\, \left( \otimes_{j=1}^n \ket{\psi_j}
\right)^{\Gamma(S_i)} \in D_i \right\}_.
$$
The cohomology ring of $\prod_{j=1}^n \PP^{d_j-1}$ is well understood:
$$H^{*}\left( \prod_{j=1}^n \PP^{d_j -1} \right) = \ZZ[\alpha_1,\cdots,\alpha_n]/(\alpha_1^{d_1}, \cdots, \alpha_n^{d_n})_.$$
A proof can be found in any textbook on algebraic topology. See \cite[Section 3.2]{hatcher} for example. By Bertini's theorem \cite[Chapter II, Theorem 8.18]{hart} again, we can choose perturbations $\PP D_i'$ of $\PP D_i$ such that $\iota^{-1} (\PP D_i' )$ are smooth and Poincar\'e dual to $(\alpha_1+\cdots +\alpha_n)^{k_i}$ for each $i=1,2,\cdots, r$. Since the complex conjugation changes the orientation, the perturbation $\left( \phi^{\Gamma(S_i)} \right)^{-1} \left( \iota^{-1} ( \PP D_i' ) \right)$ of  $\left( \phi^{\Gamma(S_i)} \right)^{-1} \left( \iota^{-1} ( \PP D_i ) \right)$  is a smooth submanifold of  $\prod_{j=1}^n \PP^{d_j-1}$, whose Poincar\'e dual is
$$
(\sigma_{i,1} \alpha_1 + \sigma_{i,2} \alpha_2 + \cdots + \sigma_{i,n} \alpha_n)^{k_i}.
$$
By the transversality theorem \cite[Chapter 3, Theorem 2.4]{hirsch} in differential topology, we can find perturbations $W_i$ in $\prod_{j=1}^n \PP^{d_j-1}$ of $\left( \phi^{\Gamma(S_i)} \right)^{-1} \left( \iota^{-1} ( \PP D_i' ) \right)$ that are still smooth and intersect transversely. Then the Poincar\'e dual of $\bigcap_{i=1}^r W_i$ is the class
$$
\prod_{i=1}^r (\sigma_{i,1} \alpha_1 + \sigma_{i,2} \alpha_2 +
\cdots + \sigma_{i,n} \alpha_n)^{k_i}
$$
in the cohomology ring $H^{*}\left( \prod_{j=1}^n \PP^{d_j -1} \right) = \ZZ[\alpha_1,\cdots,\alpha_n]/(\alpha_1^{d_1}, \cdots,
\alpha_n^{d_n})$. If the set $\bigcap_{i=1}^r \left( \phi^{\Gamma(S_i)} \right)^{-1} \left( \iota^{-1} ( \PP D_i ) \right)$ is empty, so is its small perturbation $\bigcap_{i=1}^r W_i$ and hence the cohomology class $\prod_{i=1}^r (\sigma_{i,1} \alpha_1 + \sigma_{i,2} \alpha_2 + \cdots + \sigma_{i,n} \alpha_n)^{k_i}$ should be zero. This proves the theorem.
\end{proof}

The statement (ii) of Theorem \ref{main} is an easy consequence of Theorem \ref{algcond}. Indeed, $\prod_{i=1}^r (\sigma_{i,1} \alpha_1 + \sigma_{i,2} \alpha_2 + \cdots + \sigma_{i,n} \alpha_n)^{k_i}$ in the quotient ring of the polynomial ring $\ZZ[\alpha_1,\cdots,\alpha_n]$ by the relations $\alpha_1^{d_1}= \cdots =\alpha_n^{d_n}=0$ should be a constant multiple of
$\prod_j\alpha_j^{d_j-1}$, because $N_E=\sum_i k_i=\sum_j(d_j-1)=N_U$ in the critical case.

It is worthwhile to consider the case when all the subsets $S_i\subset [n]$ are empty. In this case, $\sigma_{{i,j}}=1$ for every $i,j$ and
$$
\prod_{i=1}^r (\sigma_{i,1} \alpha_1 + \sigma_{i,2} \alpha_2 + \cdots + \sigma_{i,n} \alpha_n)^{k_i} =(\alpha_1 + \alpha_2 + \cdots + \alpha_n)^{\Sigma (d_j-1)}.
$$

It is straightforward to check that the coefficient of $\prod_j\alpha_j^{d_j-1}$ in the polynomial $(\alpha_1 + \alpha_2 + \cdots + \alpha_n)^{\sum (d_j-1)}$ is $\dfrac{\left( \sum_j (d_j -1) \right)!}{\prod_j (d_j -1)!}>0$. We thus obtain the following.

\begin{corollary}\label{segre}
Let $D_1,\cdots,D_r$ be subspaces of $\mathcal H$ with $k_i=\dim D_i^{\perp}$ for $i=1,\cdots,r$. If $\sum_{i=1}^r k_i = \sum_j(d_j-1)$, then we have the following:
\begin{enumerate}[{\rm (i)}]
\item There always exists a nonzero product vector $\ket{\psi}$ satisfying $\ket{\psi} \in D_i$ for $1 \leq i \leq r$.
\item The number of distinct nonzero product vectors $\ket{\psi}$ up to constant satisfying $\ket{\psi} \in D_i$ for $1 \leq i \leq r$ is less than or equal to $\dfrac{\left( \sum_j (d_j -1) \right)!}{\prod_j (d_j -1)!}$ if it is finite.
\item The equality holds for generic choices of $D_i$.
\end{enumerate}
\end{corollary}

We remark that (iii) was obtained in \cite[Corollary 3.9]{walgate08generic}.

%%%%%%%%%%%%%%%%%%%%%%%%%%%%%%%%%%%%%%%%%%%%%%%%%%%%%%%%%%%%%%%%%%%%%%%%%%%%%%%%%%%%%%%%%%%%%%%%%%%%%%%%%%%%%%%%%%
%%%%%%%%%%%%%%%%%%%%%%%%%%%%%%%%%%%%%%%%%%%%%%%%%%%%%%%%%%%%%%%%%%%%%%%%%%%%%%%%%%%%%%%%%%%%%%%%%%%%%%%%%%%%%%%%%%
\subsection{Under-determined case}
For the proof of the statement (iii) of Theorem \ref{main}, we introduce some vector notations. For $\bk := (\seq{k}{r})$, $\bm := (\seq{m}{n})$ and $\alpha := (\seq{\alpha}{n})$, we denote $|\bk|:=\sum k_i$, $|\bm|:=\sum m_j$ and $\alpha^{\bm} := \prod_{j=1}^n \alpha_j^{m_j}$. Let $\sigma_i:=(\sigma_{i,1},\cdots,\sigma_{i,n}) \in \{-1,+1\}^n$ so that we can write $\sigma_i\cdot \alpha:=\sigma_{i,1} \alpha_1 + \sigma_{i,2} \alpha_2 + \cdots + \sigma_{i,n} \alpha_n.$ By expanding, we write
\begin{equation*}
P^{\bk}(\alpha) := \prod_{i=1}^r (\sigma_{i,1} \alpha_1 + \sigma_{i,2} \alpha_2 + \cdots + \sigma_{i,n} \alpha_n)^{k_i} = \sum_{|\bm| = |\bk|} A^{\bk}_{\bm} \alpha^{\bm}
\end{equation*}
for $A^{\bk}_{\bm}\in \ZZ.$ For convenience, we define $A^{\bk}_{\bm}$ to be zero whenever there is a component of $\bk$ or $\bm$ which is negative. For two vectors $\bv = (\seq{v}{n})$ and $\bw=(\seq{w}{n})$ in $\ZZ^n$, we say $\bv \geq \bw$ when $v_i \geq w_i$ for all $i$. We begin with the following:

\begin{proposition}\label{underdet}
Let $D_1,\cdots, D_r$ be subspaces of $\mathcal H = \Cdn$ with $k_i=\dim D_i^{\perp}$ for $i=1,\cdots,r$. If $N_E=\sum_{i=1}^r k_i < N_U$ and $P^{\bk}(\alpha) =\prod_{i=1}^r (\sigma_{i,1} \alpha_1 + \sigma_{i,2} \alpha_2 + \cdots + \sigma_{i,n} \alpha_n)^{k_i}$ is not zero in the ring $\ZZ[\alpha]/(\alpha_j^{d_j})_{1 \leq j \leq n}$, then the system of equations {\rm (\ref{eq})} has infinitely many solutions.
\end{proposition}

\begin{proof}
Since $P^{\bk}(\alpha)$ is the Poincar\'e dual of a small perturbation of the intersection
$$
\bigcap_{i=1}^r \left( \phi^{\Gamma(S_i)} \right)^{-1} \left( \iota^{-1} ( \PP D_i ) \right) = \left\{ ([\ket{\psi_j}]) \in \prod_{j=1}^n \PP^{d_j -1}\, |\, \left( \otimes_{j=1}^n \ket{\psi_j} \right)^{\Gamma(S_i)} \in D_i \right\}
$$
as shown in the proof of Theorem \ref{algcond}, the nonvanishing of the class $P^{\bk}(\alpha)$ implies that a small perturbation of the intersection is a nonempty smooth manifold of real dimension $2(N_U - \sum_{i=1}^r k_i)>0.$ Therefore, the intersection always has infinitely many points and hence we find that there are uncountably many product vectors $\ket{\psi}$ satisfying $\ket{\psi}^{\Gamma(S_i)} \in D_i$.
\end{proof}

The next question is when  $P^\bk(\alpha)$ is nonzero in the ring $\ZZ[\alpha]/(\alpha_j^{d_j})_{1 \leq j \leq n}$. In \cite[Lemma 2]{kye-prod-vec}, it was shown that $P^\bk(\alpha)$ is always nonzero in the under-determined case if $n=2$. However it is not true even for $n=3$.

\begin{example}\label{ex_va}
Let $n=3$. Let $S_1=\{1\}$, $S_2=\{2\}$, $S_3=\{3\}$ and $S_4=\emptyset$. Let $d_1=d_2=2$, $d_3=4$, and $k_1=k_2=k_3=k_4=1$. Then $N_E=4<5=N_U.$ In the ring $\ZZ[\alpha]/(\alpha_j^{d_j})$, we have
$$
P^\bk(\alpha)=(-\alpha_1+\alpha_2+\alpha_3)(\alpha_1-\alpha_2+\alpha_3)(\alpha_1+\alpha_2-\alpha_3)(\alpha_1+\alpha_2+\alpha_3)=0
$$
since $\alpha_1^2=\alpha_2^2=\alpha_3^4=0$. Hence $P^\bk(\alpha)$ may be zero even for the under-determined case when $n=3$. We note that the associated matrix $\Sigma$ is given by
$$
\left(\begin{matrix}
-&+&+\\ +&-&+\\+&+&-\\+&+&+
\end{matrix}\right),
$$
where $+$ and $-$ denote $+1$ and $-1$, respectively.
\end{example}

In this example, the matrix $\Sigma$ has rank smaller than $r$. This suggests that we may have to impose a condition on the rank of $\Sigma$ in order to have the nonvanishing of $P^\bk(\alpha)$. Here is a criterion, and this completes the proof of (iii) of Theorem \ref{main}.

\begin{proposition}\label{pr_ncri}
Let $\Sigma=(\sigma_{i,j})$ be an $r\times n$ matrix whose entries are $\pm 1$. Let $k_1, \cdots, k_r \in \mathbb{Z}_{\geq 0}$ and $d_1, \cdots, d_n \in \mathbb{Z}_{>0}$. If $\sum_{i=1}^r k_i < \sum_{j=1}^n (d_j-1)$ and the rank of $\Sigma$ is $r$, then $P^\bk(\alpha)=\prod_{i=1}^r(\sigma_i\cdot \alpha)^{k_i}$ is nonzero in the ring $\ZZ[\alpha]/(\alpha_j^{d_j})_{1 \leq j \leq n}$ for $\bk\ge 0$.
\end{proposition}

\begin{proof}
We fix $d_1,d_2,\cdots,d_n$, $\Sigma$ and allow $\bk$ to vary. The proposition is equivalent to saying that there is an $n$-tuple of nonnegative integers $\bm := (\seq{m}{n})$ such that $|\bm| = |\bk|$, $m_j \leq d_j-1$ for every $j$ and $A^{\bk}_{\bm} \neq 0$ whenever $\sum_{i=1}^r k_i < \sum_{j=1}^n (d_j-1)$. This is obvious for $\bk=0$ since $A^{0,0,\cdots,0}_{0,0,\cdots,0}=1$. Suppose that there is a ${\bk} \geq 0$ for which the proposition fails. Let $\tilde{\bk}$ be such a vector with $|\tilde{\bk}|$ minimal.

Consider the following statement for nonnegative integers $s$ and $m$.\\

$\cT^{\bk}_{s,m}$ : \emph{All the coefficients $A^{\bk}_{\bm}$ are zero whenever $m_s=m$ or
$m_s=m+1$ and when $m_j \leq d_j-1$ for $1 \leq j \leq n$.}\\

We claim that \emph{for a fixed $\bk$ and given $s$, if the statement $\cT^{\bk}_{s,m}$ holds for some $m$, then so does the statement $\cT^{\bk-\be_i}_{s,m-1}$ for every $i$}, where $\be_i$ denotes the $i$-th standard basis vector.

This claim induces a contradiction to the minimality of $|\tilde{\bk}|$ and hence proves the proposition. Indeed, by the assumption on $\tilde{\bk}$, $\cT^{\tilde{\bk}}_{s,m}$ holds for every $s$ and $m \leq d_s-2$. Then the claim says that the statement $\cT^{\tilde{\bk}-\be_i}_{s,m}$ holds for every $s$ and $m \leq d_s-3$. In particular, $A^{\tilde{\bk}-\be_i}_{\bm}$ can be nonzero only when $m_s = d_s -1$ for every $s$, which is impossible since $|\bm| = |\tilde{\bk}-\be_i|=|\tilde{\bk}|-1 < \sum (d_j -1)=N_U$. Therefore, all the $A^{\tilde{\bk}-\be_i}_{\bm}$ are zero for every $\bm$ satisfying $|\bm| = |\tilde{\bk}|-1$ and $m_j \leq d_j-1$. This contradicts the minimality of $|\tilde{\bk}|$.

Now we prove the claim. Suppose that the statement $\cT^{\bk}_{s,m}$ holds for some $s$ and $m$. For each $j$, we take the partial derivative of $P^\bk:=P^{\bk}(\alpha)$ with respect to $\alpha_j$ to obtain the following:
$$
\frac{\partial}{\partial \alpha_j}P^{\bk} = \sum_{i=1}^r k_i \sigma_{i,j} P^{\bk - \be_i} = \sum_{i=1}^r k_i \sigma_{i,j} \left( \sum_{|\bm\pri| = |\bk - \be_i|} A^{\bk-\be_i}_{\bm\pri} \alpha^{\bm\pri} \right) = \sum_{|\bm| = |\bk|} m_j A^{\bk}_{\bm} \alpha^{\bm-\be_j}.$$

We fix an integer $\ell$. If we take the coefficient of the monomial $\alpha^{\bm-\be_{\ell}}$ of the equation above for each $j$, we get the following system of equations:

\begin{equation*}
\begin{matrix}
m_1 A^{\bk}_{\bm + \be_1 - \be_{\ell}} & =
 & k_1 \sigma_{1,1} A^{\bk -\be_1}_{\bm-\be_{\ell}}
 & + & k_2 \sigma_{2,1} A^{\bk -\be_2}_{\bm-\be_{\ell}} & + \cdots + & k_r \sigma_{r,1} A^{\bk -\be_r}_{\bm-\be_{\ell}} \\
m_2 A^{\bk}_{\bm + \be_2 - \be_{\ell}} & =
 & k_1 \sigma_{1,2} A^{\bk -\be_1}_{\bm-\be_{\ell}}
 & + & k_2 \sigma_{2,2} A^{\bk -\be_2}_{\bm-\be_{\ell}} & + \cdots + & k_r \sigma_{r,2} A^{\bk -\be_r}_{\bm-\be_{\ell}} \\
& & \cdots & & & \\
m_{{\ell}-1} A^{\bk}_{\bm + \be_{{\ell}-1} - \be_{\ell}} &
 = & k_1 \sigma_{1,{\ell}-1} A^{\bk -\be_1}_{\bm-\be_{\ell}}
 & + & k_2 \sigma_{2,{\ell}-1} A^{\bk -\be_2}_{\bm-\be_{\ell}} & + \cdots + & k_r \sigma_{r,{\ell}-1} A^{\bk -\be_r}_{\bm-\be_{\ell}} \\
m_{\ell} A^{\bk}_{\bm} & = & k_1 \sigma_{1,{\ell}} A^{\bk -\be_1}_{\bm-\be_{\ell}}
 & + & k_2 \sigma_{2,{\ell}} A^{\bk -\be_2}_{\bm-\be_{\ell}} & + \cdots + & k_r \sigma_{r,{\ell}} A^{\bk -\be_r}_{\bm-\be_{\ell}} \\
m_{{\ell}+1} A^{\bk}_{\bm + \be_{{\ell}+1} - \be_{\ell}} & =
 & k_1 \sigma_{1,{\ell}+1} A^{\bk -\be_1}_{\bm-\be_{\ell}}
  & + & k_2 \sigma_{2,{\ell}+1} A^{\bk -\be_2}_{\bm-\be_{\ell}} & + \cdots + & k_r \sigma_{r,{\ell}+1} A^{\bk -\be_r}_{\bm-\be_{\ell}} \\
& & \cdots & & & \\
m_n A^{\bk}_{\bm + \be_n - \be_{\ell}} & = & k_1 \sigma_{1,n} A^{\bk
-\be_1}_{\bm-\be_{\ell}} & + & k_2 \sigma_{2,n} A^{\bk
-\be_2}_{\bm-\be_{\ell}} & + \cdots + & k_r \sigma_{r,n} A^{\bk
-\be_r}_{\bm-\be_{\ell}}
\end{matrix}
\end{equation*}

If $m_s =m$ and $\ell \neq s$, or $m_s =m+1$ and $\ell = s$, then LHS are all zero by assumption. Therefore, we have

\begin{equation*}
\begin{pmatrix}
\text{---} & k_1 \sigma_{1,j} & \text{---} \\
\text{---} & k_2 \sigma_{2,j} & \text{---} \\
& \cdots & \\
\text{---} & k_r \sigma_{r,j} & \text{---}
\end{pmatrix}^t \cdot
\begin{pmatrix}
A^{\bk-\be_1}_{\bm-\be_{\ell}} \vspace{2mm} \\
A^{\bk-\be_2}_{\bm-\be_{\ell}} \vspace{2mm} \\ \vdots \vspace{2mm}
\\ A^{\bk-\be_r}_{\bm-\be_{\ell}}
\end{pmatrix} = O.
\end{equation*}

Since the matrix $(\sigma_{i,j})$ has rank $r$, so does the matrix $(k_i \sigma_{i,j})^t$ whenever all the $k_i \neq 0$. Hence, all the $A^{\bk-\be_i}_{\bm\pri}$ are zero for any $i$ when the $s$-th component $m'_s$ of $\bm\pri$ is $m$ and $k_i \neq 0$ for all $i$. If some $k_i$ is zero, we can simply remove the $i$-th column from the matrix $(k_i \sigma_{i,j})^t$ and $A^{\bk-\be_i}_{\bm-\be_{\ell}}$ from the column vector because $A^{\bk-\be_i}_{\bm-\be_{\ell}}=0$ by our convention. The modified matrix of $(k_i \sigma_{i,j})^t$ has full rank as well, so the column vector must be also zero. Therefore, all the $A^{\bk-\be_i}_{\bm\pri}$ are zero for every $i$ and $\bm'$ with $m'_s=m$ and $|\bm'|=|\bk|-1$.

Now, we claim that $A^{\bk-\be_i}_{\bm\pri}$ are zero for any $i$ when $m'_s=m-1$. By expanding $P^\bk(\alpha)$ directly, we obtain the following:
$$
A^{\bk}_{\bm} = \sum_{m_j
= \sum_i k_{i,j}} \prod_{i=1}^r \left( k_i \atop \bk_i \right) \prod_{j=1}^n \sigma_{i,j}^{k_{i,j}} {}_,
$$
where $\bk_i := (k_{i,1}, k_{i,2}, \cdots, k_{i,n} )$ and $\displaystyle \left( k_i \atop \bk_i \right) := \dfrac{k_i!}{\prod_j (k_{i,j}!)}$ when $k_{i,j} \geq 0$ and $k_i = |\bk_i|$. We let $\displaystyle \left( k_i \atop \bk_i \right)=0$
if some $k_{i,j}<0$. Since $\displaystyle \left( k_i \atop \bk_i \right) = \sum_j \left( k_i-1 \atop \bk_i -\be_j \right)$,
$$A^{\bk}_{\bm} = \sigma_{1,j} A^{\bk -\be_j}_{\bm-\be_1} + \sigma_{2,j} A^{\bk -\be_j}_{\bm-\be_2} + \cdots + \sigma_{n,j} A^{\bk -\be_j}_{\bm-\be_n} \quad \text{ for each }j.$$

Note that if $m_s=m$, then $A^{\bk}_{\bm}$ and $A^{\bk -\be_j}_{\bm-\be_i}$ are zero for $i \neq s$. We thus have $A^{\bk -\be_j}_{\bm-\be_s}=0$ for every $j$ and $\bm$ with $m_s=m$. Therefore, all the $A^{\bk-\be_i}_{\bm\pri}$ are zero for any $i$ when the $s$-th component of $\bm\pri$ is $m-1$ or $m$. We thus proved the statement $\cT^{\bk-\be_i}_{s,m-1}$ for every $i$. This completes the proof.
\end{proof}

In order to apply Theorem \ref{main} (iii), it helps to minimize the number $r$ in the system of equations (\ref{eq}). To do this, we may assume that the associated matrix $\Sigma$ has pairwisely non-parallel rows. Indeed, if $S_i = S_j$ (respectively $S_i = S_j^c$) for some $i \neq j$, then we can combine two systems of equations $\ket{\psi}^{\Gamma(S_i)} \in D_i$ and $\ket{\psi}^{\Gamma(S_j)} \in D_j$ into a single $\ket{\psi}^{\Gamma(S_i)} \in D_i \cap D_j$ (respectively $\ket{\psi}^{\Gamma(S_i)} \in D_i \cap \bar{D}_j$). If $r \leq 3$, then it is easy to see that pairwisely non-parallel rows of $\Sigma$ are always linearly independent. Therefore, the rank condition in Proposition \ref{pr_ncri} is automatically satisfied. This is not true for $r=4$, as we have seen in Example \ref{ex_va}.

If $n=2$, then we may always assume that $r \leq 2$ by the above argument, so the rank condition is redundant. For the $n$ qubit under-determined cases, we have $r \le N_E<N_U=n$,  so the rank condition is also redundant for the three or four qubit cases because $r \leq 3$. Therefore, we have the following. The case of $n=2$ is nothing but \cite[Theorem 3, (ii)]{kye-prod-vec}.

\begin{proposition}\label{redundant}
Let $n=2$ or $d_j=2$ with $n=3,4$. Then the system of equations  {\rm (\ref{eq})} has infinitely many solutions whenever $N_E<N_U$.
\end{proposition}

It is worthwhile to note that the converse of Proposition \ref{pr_ncri} does not hold. To see this, we consider the following two matrices in the five qubit case with $k_j=1$ for $j=1,2,3,4$:
$$
\Sigma_1= \left(\begin{matrix}
-&+&+&-&-\\
+&-&+&+&+\\
+&+&-&+&+\\
+&+&+&+&+
\end{matrix}\right)_,
\qquad \Sigma_2= \left(\begin{matrix}
-&+&+&-&+\\
+&-&+&+&-\\
+&+&-&+&+\\
+&+&+&+&+
\end{matrix}\right)_.
$$
These are of rank three. It is interesting to note that $P^\bk(\alpha)=0$ for $\Sigma_1$, but $P^\bk(\alpha)$ is nonzero for $\Sigma_2$ in the ring $\ZZ[\alpha_1, \cdots, \alpha_5]/(\alpha_1^2, \cdots, \alpha_5^2)$. Therefore, the converse of Proposition \ref{pr_ncri} does not hold.

In the trivial case where $S_i \subset [n]$ are all empty or $[n]$, $P^{\bk}(\alpha) = \pm (\alpha_1 + \cdots + \alpha_n)^{N_E}$ is always nonzero because $N_E < N_U$ and $(\alpha_1 + \cdots + \alpha_n)^{N_U} \neq 0$ in $\ZZ[\alpha]/(\alpha_j^{d_j})$ by Corollary \ref{segre}. By Proposition \ref{underdet}, the system of equations \eqref{eq} has infinitely many nonzero solutions for any $D_i$ with $\dim D_i^{\perp}=k_i$.

%%%%%%%%%%%%%%%%%%%%%%%%%%%%%%%%%%%%%%%%%%%%%%%%%%%%%%%%%%%%%%%%%%%%%%%%

\section{Multi-qubit cases and Permanents of matrices}

In this section, we investigate the multi-qubit cases where $d_j=2$ for all $j$ so that $N_U=n$. In the critical case where the numbers of equations $N_E$ and unknowns $N_U$ coincide in the system of equations (\ref{eq}), we may assume that $k_i=1$ for all $i$ because if $k_i >1$ we can repeat $S_i$ $k_i$ times and replace $D_i$ by $k_i$ hyperplanes. In particular, we may assume $N_E=r=n=N_U$. By Theorem \ref{main} (ii), the solvability of (\ref{eq}) is guaranteed by the nonvanishing of the coefficient of the monomial $\alpha_1\alpha_2\cdots\alpha_n$ in the polynomial (\ref{criti}), which is
\begin{equation}\label{perm}
\sum_{\lambda\in \Sym(n)} \sigma_{1,\lambda(1)}\sigma_{2,\lambda(2)}\cdots \sigma_{n,\lambda(n)},
\end{equation}
where $\Sym(n)$ denotes the set of all permutations of the set $[n]$. If we multiply the sign of permutation in each summand, this is nothing but the determinant of the matrix $\Sigma=[\sigma_{i,j}]$. The number (\ref{perm}) is called the \emph{permanent} of the matrix $\Sigma$, which has been studied since Cauchy's era. See the monograph \cite{minc}. The permanent of $\Sigma$ will be denoted by $\per(\Sigma)$. By Theorem \ref{main} (ii), we have the following:

\begin{theorem}\label{qubit}
Let $\{S_1,\cdots,S_n\}$ be subsets of $[n]$ with the associated $n\times n$ matrix $\Sigma$, and $\{D_1,\cdots, D_n\}$ subspaces of $\bigotimes_{i=1}^n\mathbb C^2$ with $\dim D_i^\perp=1$ for $i=1,\cdots, n$, respectively. If $\per (\Sigma)\neq 0$ then the system of equations {\rm (\ref{eq})} has a nonzero solution.
\end{theorem}

Therefore, in order to check the existence of a nonzero solution of (\ref{eq}) for the $n$ qubit cases with the same numbers of equations and unknowns, we have to calculate the permanents of the associated matrices whose entries are $\pm 1$. Several authors have studied permanents of those matrices. It was shown in \cite{wang74} that if $n \geq 2$ is even or $n \equiv 1 (\text{mod } 4)$, then there exists an $n \times n$ $(+1,-1)$-matrix $A$ with $\per(A)=0$. In the same paper, it was also noticed that there is no $3 \times 3$ $(+1,-1)$-matrix with vanishing permanent. It was proved in \cite{krauter,simion83,per05} that there exists an $n \times n$ $(+1, -1)$-matrix with vanishing permanent if and only if $n+1$ is not a power of $2$. Therefore, we have the following:

\begin{theorem}\label{37qubit}
Let $n=2^k-1$ for $k=2,3,\cdots$ and $d_i=2$ for $i=1,2,\cdots,n$. Then the system of equations {\rm (\ref{eq})} has a nonzero solution whenever the number of equations are less than or equal to $n$.
\end{theorem}

The above theorem does not hold even for the two qubit case with $n=2$, as it was discussed in \cite{kye-prod-vec}. We recall the following typical example. Let $\{ \ket{0}, \ket{1} \}$ be an orthonormal basis for $\CC^2$. Let $\ket{\beta_1} = \ket{00} + \ket{11}$ and $\ket{\beta_2} = \ket{01} - \ket{10}$. For two nonzero vectors $\ket{\psi_1}, \ket{\psi_2}$ in $\CC^2$, we have
$$
\braket{\psi_1 ,\bar{\psi}_2}{\beta_1}
= \braket{\psi_1}{0}\braket{\bar{\psi}_2}{0} + \braket{\psi_1}{1}\braket{\bar{\psi}_2}{1}
=\braket{\psi_1}{\psi_2}.
$$

Therefore, we see that the equation $\braket{\psi_1 , \bar{\psi}_2}{\beta_1} =0$ is equivalent to the orthogonality of $\ket{\psi_1}$ and $\ket{\psi_2}$. Similarly, the equation $\braket{\psi_1 ,\psi_2}{\beta_2} =0$ is equivalent to saying that $\ket{\psi_1}$ and $\ket{\psi_2}$ are parallel. If we put $D_1=\ket{\beta_1}^\perp$ and $D_2=\ket{\beta_2}^\perp$ then the system of equations
$$
\begin{aligned}
\ket{\psi_1, \bar{\psi}_2}\in D_1\\
\ket{\psi_1, {\psi}_2}\in D_2
\end{aligned}
$$
has no nonzero solution. Note that the associated matrix is given by
$$
\left(\begin{matrix}
+&-\\ +&+\end{matrix}\right)
$$
with vanishing permanent. We modify this example to get the same kind of a system of equations for the four qubit case with the same number of equations and unknowns.

\begin{example}\label{4_qu}
Let subspaces $\{D_1,D_2,D_3,D_4\}$ of $\bigotimes_{j=1}^4\mathbb C^2$ be given by
$$
D_1 = (\ket{\beta_1}\ts \ket{\beta_1})\orth,\
D_2 = (\ket{\beta_1}\ts \ket{\beta_2})\orth,\
D_3 = (\ket{\beta_2}\ts \ket{\beta_1})\orth,\
D_4 = (\ket{\beta_2}\ts \ket{\beta_2})\orth.
$$
Then, we have
\begin{equation}\label{eq:4qubit}
\begin{split}
\ket{\psi_1, \bar{\psi}_2} \ts \ket{\psi_3, \bar{\psi}_4} \in D_1 \quad
   &\Longleftrightarrow \quad \ket{\psi_1} \perp \ket{\psi_2} \quad \text{or} \quad \ket{\psi_3} \perp \ket{\psi_4} \\
\ket{\psi_1, \bar{\psi}_2} \ts \ket{\psi_3, \psi_4} \in D_2 \quad
   &\Longleftrightarrow \quad \ket{\psi_1} \perp \ket{\psi_2} \quad \text{or} \quad \ket{\psi_3} \parallel \ket{\psi_4} \\
\ket{\psi_1, \psi_2} \ts \ket{\psi_3, \bar{\psi_4}} \in D_3 \quad
   &\Longleftrightarrow \quad \ket{\psi_1} \parallel \ket{\psi_2} \quad \text{or} \quad \ket{\psi_3} \perp \ket{\psi_4} \\
\ket{\psi_1, \psi_2} \ts \ket{\psi_3, \psi_4} \in D_4 \quad
   &\Longleftrightarrow \quad \ket{\psi_1} \parallel \ket{\psi_2} \quad \text{or} \quad \ket{\psi_3} \parallel \ket{\psi_4}
\end{split}
\end{equation}

It is clear that there exists no nonzero product vector $\ket{\psi_1, \psi_2, \psi_3, \psi_4}\in\bigotimes_{j=1}^4\mathbb C^2$ satisfying all of these equations. Note that the the associated matrix is
$$
\begin{pmatrix}
+ & - & + & - \\ + & - & + & + \\ + & + & + & - \\ + & + & + & +
\end{pmatrix}_,
$$
which has the vanishing permanent. If we take the last three columns then it is equivalent to the associated matrix in Example \ref{ex_va}. Employing the above method to construct the example for $n=4$ from the example for $n=2$, it is easy to construct the same kind of examples when $n=2^k$ for $k=3,4,\cdots$.
\end{example}

We say that two $r\times n$ matrices $\Sigma_1$ and $\Sigma_2$ are equivalent if $\Sigma_2$ is obtained from $\Sigma_1$ by a succession of the following operations:
\begin{enumerate}[(i)]
\item
interchange two rows or columns,
\item
negate a row or a column.
\end{enumerate}

Interchanging two rows and columns is equivalent to changing the orders of equations and unknowns in (\ref{eq}), and negating a row or a column is equivalent to conjugating an equation or an unknown in (\ref{eq}). Therefore, two systems of equations like (\ref{eq}) have the same solvability if their associated matrices are equivalent.

It is a natural problem to classify all $n\times n$ $(+1, -1)$-matrices with vanishing permanents, up to equivalence. The first step for classification is to reduce the number $\mu(\Sigma)$ of minus signs, that is, the number of $-1$'s in the entries of $\Sigma$. We also denote by $r_i(\Sigma)$ (respectively $c_j(\Sigma)$) the number of minus signs in the $i$-th row (respectively the $j$-th column) of $\Sigma$.

\begin{proposition}\label{vanish}
Suppose that $n\ge 3$. For a given $n\times n$ matrix $\Sigma=[\sigma_{ij}]$ with entries $\pm 1$, we have the following:
\begin{enumerate}
\item[{\rm (i)}]
If $n=2m+1$ is an odd number and $\mu(\Sigma)\ge mn-(m-1)$, then there exists $\Sigma^\prime$ which is equivalent to $\Sigma$ such that $\mu(\Sigma^\prime)<\mu(\Sigma)$.
\item[{\rm (ii)}]
If $n=2m$ is an even number and $\mu(\Sigma)\ge mn-m$, then there exists $\Sigma^\prime$ which is equivalent to $\Sigma$ such that $\mu(\Sigma^\prime)<\mu(\Sigma)$.
\end{enumerate}
\end{proposition}

\begin{proof}
If there is a column with $m+1$ minus signs then we may decrease the number $\mu(\Sigma)$ strictly by negating this column, and the same for rows. Therefore, it remains to consider the case when all the columns and rows have at most $m$ minus signs. Put
$$
I=\{i\in[n]: r_i(\Sigma)=m\},\qquad J=\{j:\in [n]: c_j(\Sigma)=m\}.
$$

We note that if $|I|\le\ell$ then
$$
\mu(\Sigma)\le m\cdot|I|+(m-1)(n-|I|)=mn-n+|I|\le mn-n+\ell,
$$
and the same for $J$, where $|I|$ denotes the cardinality of $I$. Therefore, we have
$$
\mu(\Sigma)\ge mn-n+\ell\ \Longrightarrow\ |I|\ge\ell,\ |J|\ge \ell.
$$

In case of (i), we have $mn-(m-1)=mn-n+(m+2)$, and so it follows that $|J|\ge m+2$ by assumption. Therefore, for any $i\in I$, there exist at least two $j\in J$, say $\{j_1,j_2\}$, with $\sigma_{ij}=+1$. Take any $i\in I$ and negate the $i$-th row, to get $\Sigma^\pr$ with $\mu(\Sigma^\pr)=\mu(\Sigma)+1$. If we negate the $j_1$-th and $j_2$-th columns to get $\Sigma^{\pr\pr}$, then we have $\mu(\Sigma^{\pr\pr})\le \mu(\Sigma^\pr)-2 =\mu(\Sigma)-1$.

In the even case $n=2m$, we first consider the case $\mu(\Sigma)\ge mn-(m-1)=mn-n+(m+1)$. In this case, we have $|J|\ge m+1$, and so for any $i\in I$ there exists at least one $j\in J$ with $\sigma_{ij}=+1$. We apply the same argument as in the odd $n$ case, to get $\Sigma^\pr$ and $\Sigma^{\pr\pr}$. In this case, we have $\mu(\Sigma^{\pr\pr})\le \mu(\Sigma^\pr)-1=\mu(\Sigma)-1$.

It remains to prove when $n=2m$ and $\mu(\Sigma)= mn-m$, which implies $|I|\ge m$ and $|J|\ge m$. In this case, we consider the set $I\times J$. If there exists $(i,j)\in I\times J$ with $\sigma_{ij}=+1$ then negate the $i$-row and the $j$-th column, to get the conclusion. If $\sigma_{ij}=-1$ for each $(i,j)\in I\times J$ then we see that $|I|=|J|=m$. In this case we negate the $i$-th row for each $i\in I$ to get $\Sigma^\prime$. Then there exist $j\in [n]\setminus J$ such that $c_j(\Sigma^\pr)>m$ since $\mu(\Sigma)>|I\times J|$ by the assumption $n\ge 3$. Negate this column to get the required conclusion.
\end{proof}

When $n$ is a power of $2$, the following proposition is also useful for classification of $(+1,-1)$-matrices with vanishing permanents. We recall the following addition formula for permanents:

$$
\per(A+B) = \sum_{i=0}^n \sum_{{S,T \in [n]} \atop {|S|=|T|=i}} \per(A[S|T]) \per(B(S|T)),
$$
where $A[S|T]$ is the submatrix of $A$ consisting of rows indexed by $S$ and columns indexed by $T$, and $B(S|T)$ is the submatrix of $B$ deleting rows indexed by $S$ and columns indexed by $T$. If $|S|=|T|=0$ (respectively $|S|=|T|=n$), we set $\per(A[S|T])=1$ (respectively $\per(B(S|T))=1$). See \cite[Chapter 2, Theorem 1.4]{minc}. This formula holds for arbitrary $n \times n$ matrices $A$ and $B$.

\begin{proposition}\label{even}
Suppose that $n=2^k$ for $k=2,3,\cdots$. If an $n\times n$ matrix $\Sigma$ with entries $\pm 1$ has the vanishing permanent, then $\mu(\Sigma)$ must be even.
\end{proposition}

\begin{proof}
We write the $n \times n$ matrix $\Sigma = [ \sigma_{i,j} ]$ as $J-2P$ where $J$ is the matrix whose entries are all $+1$ and $P$ is a uniquely determined matrix whose entries are $0$ or $+1$. By the addition formula, we obtain the formula
$$
\per(\Sigma) = \sum_{i=0}^n (-2)^i (n-i)!\, \per_i(P),
$$
where $\per_i (P)$ is the sum of all permanents of $i \times i$ submatrices of $P$. See \cite{simion83}. The largest natural number $N_i$ such that $2^{N_i}$ divides the $i$-th summand $(-2)^i (n-i)!$ is given by
$$
N_i=\begin{cases}
n-1,\quad &i=0,\\
n-k, &i=1,\\
i+ \sum_{j=1}^k \floor{\frac{n-i}{2^j}}
&i=2,3,\cdots,n,
\end{cases}
$$
where $\floor{x}$ is the largest integer which is not greater than $x$. We show that $N_i\ge n-k+1$ for $i=2,3,\cdots,n$. Let $n-i = a_{k-1} 2^{k-1} + a_{k-2} 2^{k-2} + \cdots + a_0$ be the $2$-adic expansion of $n-i$. Then we have $\sum_{j=0}^{k-1} a_j\le k-1$, because some of $a_i$ must be zero by $i\ge 2$. It is easy to see
$$
\sum_{j=1}^k \floor{\dfrac{n-i}{2^j}}=n-i-\sum_{j=0}^{k-1}a_j.
$$
Therefore, we have $N_i=n-\sum_{j=0}^{k-1}a_j \ge n-k+1$, and so
$$
\per\Sigma
\equiv (-2)(n-1)! \cdot \per_1(P)
\equiv 2^{n-k}\ell\cdot \per_1(P) \mod 2^{n -k +1},
$$
where $\ell$ is an odd number. Since $\per\Sigma=0$, we see that $\mu(\Sigma)=\per_1(P)$ must be an even number.
\end{proof}

In order to classify $4\times 4$ $(+1,-1)$-matrices with vanishing permanents up to equivalence, we may consider only the cases $\mu=2$ and $\mu=4$, by Propositions \ref{vanish} and \ref{even}. In the case of $\mu=2$, one can check that we have only two permanent vanishing matrices up to equivalence:
$$
\Sigma_1
=\left(\begin{matrix}
-&-&+&+\\
+&+&+&+\\
+&+&+&+\\
+&+&+&+
\end{matrix}\right)_,
$$
and its transpose $\Sigma_1^\ttt$.

In the case of $\mu=4$, we have to investigate the following cases:
\begin{enumerate}[(i)]
\item there are two rows with two $-1$'s,
\item there are one row with two $-1$'s and two rows with one $-1$,
\item there are four rows with one $-1$.
\end{enumerate}
In the case of (i), there is only one matrix with vanishing permanent up to equivalence:
$$\Sigma_2 = \left(\begin{matrix}
-&-&+&+\\
+&-&-&+\\
+&+&+&+\\
+&+&+&+
\end{matrix}\right).$$
In the case of (ii), there are only three matrices with vanishing permanents up to equivalence:
$$\Sigma_2^\ttt
= \left(\begin{matrix}
-&+&+&+\\
-&-&+&+\\
+&-&+&+\\
+&+&+&+
\end{matrix}\right), \qquad
\Sigma_3
=\left(\begin{matrix}
-&-&+&+\\
+&-&+&+\\
+&+&-&+\\
+&+&+&+
\end{matrix}\right),\qquad
\Sigma_4
=\left(\begin{matrix}
-&-&+&+\\
+&+&-&+\\
+&+&+&-\\
+&+&+&+
\end{matrix}\right).
$$
We note that $\Sigma_2^\ttt$ is equivalent to the associated matrix in Example \ref{4_qu} and the transpose of $\Sigma_3$ is equivalent to $\Sigma_3$ itself. In the case of (iii), there is only one matrix $\Sigma_4^\ttt$ with vanishing permanent. If we negate the first row of $\Sigma_4$ and rearrange the rows and columns appropriately, then we get the matrix $\Sigma_2^\ttt$. Therefore, $\Sigma_4$ is equivalent to $\Sigma_2^\ttt$. This implies that $\Sigma_4^\ttt$ is also equivalent to $\Sigma_2$. To summarize, we have at most five inequivalent $(+1,-1)$-matrices with vanishing permanents up to equivalence:

$$
\Sigma_1, \quad \Sigma_1^\ttt, \quad \Sigma_2, \quad \Sigma_2^\ttt, \quad \Sigma_3.
$$

We claim that these five matrices are inequivalent. Since $\Sigma_i$ and $\Sigma_i^t$ have rank $i+1$ for $i=1,2,3$, we find that neither $\Sigma_i$ nor $\Sigma_i^\ttt$ is equivalent to $\Sigma_j$ or $\Sigma_j^\ttt$ if $i \neq j$. It remains to show that $\Sigma_1$ (respectively $\Sigma_2$) and $\Sigma_1^\ttt$ (respectively $\Sigma_2^\ttt$) are not equivalent.

In other to get another invariant to distinguish them, we consider the difference $\pi_r(\Sigma)$ of the two numbers $|\{i\in [n]: r_i(\Sigma)\ {\text{\rm is even}}\}|$ and $|\{i\in [n]: r_i(\Sigma)\ {\text{\rm is odd}}\}|$ for an $n\times n$ matrix $\Sigma$ with entries $\pm 1$. If $n$ is even then it is easily checked that the number $\pi_r(\Sigma)$ is an invariant under the equivalence relation. The number $\pi_c(\Sigma)$ may be defined for columns in the same way. Since $\pi_r(\Sigma_1) = 4$ and $\pi_r(\Sigma_1^\ttt) = 0$, $\Sigma_1$ and $\Sigma_1^\ttt$ are not equivalent. Similarly, we also check $\pi_r(\Sigma_2) = 4$ and $\pi_r(\Sigma_2^\ttt) = 0$, to confirm that  $\Sigma_2$ is not equivalent to $\Sigma_2^\ttt$.

\begin{theorem}
There exist exactly five $4\times 4$ $(+1,-1)$-matrices $\Sigma_1, \Sigma_1^\ttt, \Sigma_2,\Sigma_2^\ttt,\Sigma_3$ with vanishing permanents, up to the equivalence relation.
\end{theorem}

We have considered the rank and the invariant $\pi_r(\Sigma)$ to classify permanent vanishing $(+1,-1)$-matrices in the $4 \times 4$ cases. The absolute values of the determinant and permanent are also obvious invariants under the equivalence relation. The following example shows that these do not constitute a complete set of invariants.

\begin{example}
Consider the following two matrices:
$$A= \begin{pmatrix}
- & - & + & + \\ + & - & - & + \\ - & + & - & + \\ + & + & + & +
\end{pmatrix}_,\quad B= \begin{pmatrix}
+ & + & + & + \\ + & - & + & - \\ + & + & - & - \\ + & - & - & +
\end{pmatrix}_.$$
We can check that
$$
\begin{aligned}
&\per(A) = \per(B) = 8,\\
&|\det(A)|=|\det(B)|=16,\\
&\rk(A)=\rk(B)=4,\\
&\pi_r(A)=\pi_r(B)=\pi_c(A)=\pi_c(B)=4.
\end{aligned}
$$
Note that $B B^t = 4 I_4$, where $I_4$ is the $4 \times 4$ identity matrix. It is easy to see that if $B\pri$ is equivalent to $B$, then $B\pri {B\pri}^t$ is also $4$ times the identity matrix. Since $A A^t \neq 4 I_4$, $A$ is not equivalent to $B$.
\end{example}

We close this section by mentioning an interesting asymptotic result on permanents by  Tao and Vu \cite{tao}. For the $n\times n$ matrix $M_n$ whose entries are independent and identically distributed random variables taking values $\pm 1$ with probability $1/2$ for each, they showed that asymptotically almost surely, the absolute value of $\per(M_n)$ is $n^{(\frac 12+ o(1))n}$. In particular, the probability that $\per(M_n)= 0$ tends to $0$, as $n\to\infty$.

%%%%%%%%%%%%%%%%%%%%%%%%%%%%%%%%%%%%%%%%%%%%%%
%%%%%%%%%%%%%%%%%%%%%%%%%%%%%%%%%%%%%%%%%%%%%%%%%%%%%%%%%%%%%%%%%%%%%%%%%%%%%%%%%%%%%%%%%%%%%%%%%%%%%%%%%%%%%%%%%%
\section{PPT entangled edge states and related questions}

An $n$-partite PPT entangled state is said to be an \emph{edge state} if there exists no nonzero product vector $\ket\psi$ such that $\ket{\psi}^{\Gamma(S)}\in{\mathcal R}(\varrho^{T(S)})$ for every subset $S$ of $[n]$, where ${\mathcal R}(\varrho)$ denotes the range of $\varrho$. Edge states play an important role in understanding the structure of the convex set of all PPT states, because every PPT state is the convex combination of a separable state and an edge state. Furthermore, every extreme point of the convex set of all PPT states must be a pure product state or an edge state.

The first step to classifying PPT entangled edge states is to  consider ranks of them and their partial transposes. In the $d_1\ot d_2\ot\cdots d_n$ system, we have to consider $2^{n-1}$ subsets $\{S_i\}$ of $[n]$, as it was discussed in Introduction. In this case, we have to solve the system of equations
\begin{equation}\label{edge}
\ket{\psi}^{\Gamma(S_i)}\in{\mathcal R}(\varrho^{T(S_i)}),\qquad i=1,2,\cdots, 2^{n-1},
\end{equation}
in order to check if a given PPT state $\varrho$ is an edge state or not. We first consider the following statement:

\begin{enumerate}
\item[(E${}_1$)]
If the number of equations $N_E$ is less than or equal to the number of unknowns $N_U$ in the system of equations (\ref{eq}), then there exists a nonzero solution.
\end{enumerate}
The validity of the statement (E${}_1$) for the bi-partite case has been discussed in \cite{kye-prod-vec}. It is related with Diophantine equations arising from the Krawtchouk polynomials, which play a role in the coding theory \cite{MWS,vint}. It is not yet solved completely. See also \cite[Section 7]{kye_ritsu}. Theorem \ref{37qubit} tells us that (E${}_1$) is true for the $n$ qubit systems when $n=2^k-1$ with $k=2,3,\cdots$. On the other hand, Example \ref{4_qu} shows that (E${}_1$) does not hold for $n$ qubits with $n=2^k$.

We note that the number of equations and unknowns in (\ref{edge}) are given by
$$
\sum_{i=1}^{2^{n-1}}\left(\prod_{j=1}^n d_j-\rk \varrho^{\Gamma(S_i)}\right) \quad
{\text{\rm and}}\quad
\sum_{j=1}^n(d_j-1),
$$
respectively. Therefore, the statement (E${}_1$) implies the following:

\begin{enumerate}
\item[(E${}_2$)]
If $\varrho$ is a PPT entangled edge state then we have
\begin{equation}\label{edge-ine}
\sum_{i=1}^{2^{n-1}}\rk \varrho^{T(S_i)}< 2^{n-1}\prod_{j=1}^n d_j-\sum_{j=1}^n(d_j-1).
\end{equation}
\end{enumerate}
In the bi-partite case $M\ot N$, we have the inequality
$$
\rk\varrho+\rk\varrho^T < 2MN-M-N+2,
$$
where $\varrho^T$ denotes the partial transpose of the bi-partite state $\varrho$. In the $2\ot 2$ and $2\ot 3$ systems, the statement (E${}_2$) is vacuously true by the Woronowicz-Horodecki criterion \cite{horo-1,woronowicz} which says that every PPT state must be separable in these cases. The $2\ot 2$ case goes back to the St\o rmer's work \cite{stormer} in the sixties to classify extremal positive maps between $M_2$, together with the duality  \cite{eom-kye,horo-1,woronowicz} between positive maps and bi-partite entanglement. We note that the statement (E${}_1$) is false for the $2\ot 2$ case, as it was discussed in the last section. The validity of (E${}_2$) is still open for the $2\ot 4$ case. See \cite[Section 7]{kye_ritsu}. In the $3\ot 3$ system, the statement (E${}_2$) is false. Actually, $3\ot 3$ PPT entangled edge states $\varrho$ have been constructed in \cite{kye_osaka} with $\rk\varrho=8$ and $\rk\varrho^T=6$.

On the other hand, we see that the inequality (\ref{edge-ine}) becomes
$$
\sum_{i=1}^{2^{n-1}}\rk\varrho^{T(S_i)}< 2^{2n-1}-n,
$$
for the $n$ qubit cases, and get the inequality (\ref{3qu}) for the three qubit case. We do not know if this is true for the $2^k$ qubit cases, even though the statement (E${}_1$) is false in these cases by Example \ref{4_qu}. We summarize in Table 1.

\renewcommand{\arraystretch}{1.5}
\begin{table}
\begin{center}
\begin{tabular}{|c|c|c|c|c|c|c|c|}
\hline
& \multicolumn{4}{|c|}{bi-partite cases ($k \geq 2$)} &
\multicolumn{3}{|c|}{$n$-qubit cases ($k \geq 2$)} \\\cline{2-8}
& $2 \ts 2$ & $2 \ts 2k$ & $2 \ts (2k-1)$ & $3 \ts 3$ & $n=2^k-1$ & $n=2^k$ & otherwise\\\hline
($E_1$) & No & No & Yes & No & Yes & No &?\\\hline
($E_2$) & Yes & ? & Yes & No & Yes & ? &?\\\hline
\end{tabular}
\vspace{3mm}
\caption{This table shows validities of the statements ($E_1$) and ($E_2$) in various cases. Especially, the statements for \lq otherwise\rq\ cases, including the five qubit case, are completely untouched.}
\end{center}
\end{table}

The system of equations (\ref{eq}) with complex unknowns and their conjugates is essentially a system of real equations due to the conjugation of complex numbers, and it makes the problem delicate. We could not give a definite answer even for the under-determined case when the number of equations is strictly less than the number of equations. So, we ask:

\begin{question}
Is it possible to remove the rank condition in Theorem \ref{main} (iii) ?
\end{question}

Considering Example \ref{ex_va}, this is a part of the more fundamental question, which was conjectured affirmatively in \cite{kye-prod-vec} for the bi-partite case of $n=2$. The rank condition is also redundant for the three and four qubit cases as shown in Section 2.

\begin{question}
Is the converse of Theorem \ref{algcond} true? More precisely, can one find $k_i, d_j$ and $\sigma_{i,j}$ such that $P^{\bk}(\alpha)=0$ in $\ZZ[\alpha]/(\alpha_j^{d_j})$ and that a nonzero solution of (\ref{eq}) exists for every subspace $D_i$ with $\dim D_i^{\perp}=k_i$?
\end{question}

We found unexpected relations between the existence of nonzero solutions and the permanents of $(+1,-1)$-matrices. It is obvious that the permanent is invariant under taking the transpose. Therefore, it is tempting to add the operation of transpose, for the definition of equivalence for $(+1,-1)$-matrices. But, we could not determine if the solvability of \eqref{eq} is invariant under transpose.

\begin{question}
Is the existence of nonzero solutions for (\ref{eq}) with an associated matrix $\Sigma$ equivalent to that with the associated matrix $\Sigma^\ttt$?
\end{question}

The next obvious question is to classify $(+1,-1)$-matrices up to the equivalence relation. This must be very hard in general, because it involves the word problem.

\begin{question}
Find a complete set of invariants to distinguish $n\times n$ $(+1,-1)$-matrices with vanishing permanents, up to the equivalence relation.
\end{question}

We could answer this question for $n\le 4$. We found five inequivalent $4\times 4$ $(+1,-1)$-matrices with vanishing permanents. But, we could not decide if there exists a system of equations without nonzero solutions with these associated matrices, except for the case of $\Sigma_2^\ttt$.

\begin{question}
For $\Sigma=\Sigma_1,\Sigma_1^\ttt,\Sigma_2,\Sigma_3$, is it possible to construct equation (\ref{eq}) with the associated matrix $\Sigma$ which has no nonzero solution?
\end{question}

One of our main motivations for this study was to understand the inequality (\ref{edge-ine}) for PPT entangled edge states, as it was proposed in \cite{abls} for the three qubit case. It is an interesting problem to fill up Table 1.

\end{document}